\documentclass[sigplan,10pt,nonacm,balance=false]{acmart}
\setcopyright{none}
\renewcommand\footnotetextcopyrightpermission[1]{}
\usepackage{amsmath,booktabs,tabularx,multirow}
\ifPDFTeX\else\usepackage{amssymb}\fi
\usepackage{tikz}
\usetikzlibrary{arrows.meta,positioning,fit,calc,backgrounds}
\usepackage{graphicx,microtype,placeins,needspace,pifont}
\newcommand{\latchcheckmark}{\text{\ding{51}}}
\newcommand{\sys}{TLSLatch}
\newcommand{\batch}{TLSLatch-Batch}
\newcommand{\Hash}{\mathsf{H}}

\newcommand{\dc}{\ensuremath{\latchcheckmark}}
\newcommand{\op}{\ensuremath{\circ}}

\definecolor{navy}{HTML}{234B71}
\definecolor{teal}{HTML}{176B62}
\definecolor{brick}{HTML}{A44335}
\definecolor{pale}{HTML}{F1F4F6}
\tikzset{>=Latex,every node/.style={font=\normalfont\normalsize},endpoint/.style={draw=navy,rounded corners=2pt,fill=pale,align=center,inner sep=5pt},boundary/.style={draw=teal,thick,dashed,rounded corners=3pt,inner sep=9pt},tls/.style={->,thick,navy},evidence/.style={->,dashed,teal,thick}}
\theoremstyle{definition}
\newtheorem{definition}{Definition}
\newtheorem{observation}{Observation}
\newtheorem{lemma}{Lemma}
\newtheorem{proposition}{Proposition}

\AtBeginDocument{\raggedbottom}

\title[Session Attestation for Unmodified TLS Services]{Session Attestation for Unmodified TLS Services in Confidential Virtual Machines}
\author{Qi Gu}
\affiliation{\institution{NSFOCUS Technologies Group Company Ltd}\country{China}}
\email{guqi@nsfocus.com}

\author{Wenmao Liu}
\affiliation{\institution{NSFOCUS Technologies Group Company Ltd}\country{China}}
\email{liuwenmao@nsfocus.com}

\author{Weijing You}
\affiliation{\institution{Fujian Normal University}\country{China}}
\email{youweijing@fjnu.edu.cn}

\author{Yifei Chen}
\affiliation{\institution{Hefei University of Technology}\country{China}}
\email{yifeichen@mail.hfut.edu.cn}

\author{Sheng Ma}
\affiliation{\institution{Independent Researcher}\country{China}}
\email{tianma666999@gmail.com}

\author{Fozhong Chen}
\affiliation{\institution{NSFOCUS Technologies Group Company Ltd}\country{China}}
\email{chenfozhong@gmail.com}

\begin{document}
\begin{abstract}
Confidential cloud services aim to protect sensitive requests from the infrastructure that executes them. However, running a service inside a trusted execution environment does not ensure that users’ plaintext appears only within the protected environment. We formulate Endpoint-Substitution Relay (ESR), a common attack outcome in which an adversary receives plaintext at a client-accepted endpoint while relaying requests to the legitimate service to preserve correct behavior. We present TLSLatch, a transparent session-attestation mechanism for services running in confidential virtual machines. TLSLatch attests the protected origin of the server’s ephemeral TLS 1.3 key share and gates outbound traffic until verification succeeds. It requires no changes to applications, TLS libraries, certificates, or application protocols, and adds no extra payload-encryption layer. We implement TLSLatch with a hardware-backed Hygon CSV CVM server and clients on Linux, Windows, and macOS. Across these platforms, TLSLatch reduces completion time for 1~KB requests by 56.9\%--65.5\% compared with nested TNG, and for 64~MB requests by 43.8\%--82.4\% compared with CMaaS using application-key reuse. These results show that transparent session attestation can preserve existing TLS stacks while adding low-overhead endpoint binding to confidential cloud services.
\end{abstract}
\maketitle
% !TeX root = ../main.tex
\section{Introduction}
\label{sec:intro}
Confidential computing protects data in use from privileged infrastructure. Confidential virtual machines (CVMs) provide this protection by isolating the guest OS and its applications from the untrusted host~\cite{amdsevsnp,inteltdx}. This model enables cloud services to process sensitive inputs while reducing trust in the underlying infrastructure. It is increasingly relevant to private cloud compute (PCC), including confidential AI inference services such as Alibaba Cloud’s Confidential MaaS (CMaaS)\cite{aliyuncmaas}. Such services require not only protected execution, but also assurance that sensitive data reaches the intended protected environment.

Moving a service into a CVM protects its execution, but does not by itself determine where a client’s TLS connection terminates. Requests may still traverse gateways, load balancers, and other network infrastructure before reaching the application. TLS authenticates the service identity and protects the connection~\cite{rfc8446}, while remote attestation (RA) provides evidence about the execution environment~\cite{rfc9334}. A client of a confidential service must connect these checks: the environment it approves must also be the environment that receives its plaintext. A valid TLS identity and valid attestation evidence are insufficient if they authenticate different endpoints.

Operational credentials and routing controls can be compromised independently of the production CVM. An adversary with such access may redirect a client to another endpoint while the genuine service continues to operate correctly. If that endpoint can read the request and relay it to the genuine service, the client may still receive the expected result even though confidentiality has already been lost. Prior work has shown that compromising reusable credentials can enable relay or substitution attacks against some TLS-attestation constructions~\cite{weinhold2025tlsra,sardar2026identity}. Security therefore depends not only on performing both checks, but on binding them to the same receiving endpoint.

We study this problem through \emph{Endpoint-Substitution Relay} (ESR), a common attack goal in which the client accepts a receiver whose plaintext is readable to the adversary, while the genuine service is used to preserve the expected behavior. ESR exposes the resources on which different channel-authentication constructions depend. The analysis develops two sufficient constructions for obtaining such a receiver and characterizes the service-identity, protected-environment, and receiving-state resources they require, together with the reuse scope of compromised resources. This framework provides a common basis for comparing TLS--attestation designs under different compromise conditions.

The analysis motivates a design that binds attestation to receiving state specific to the current connection while retaining independent service identity authentication. Existing approaches that provide comparable session binding or endpoint protection may require changes to the TLS or application path, or retain an additional payload-encryption layer~\cite{weinhold2025tlsra,tngsoftware,cmaassdk}. We present \sys{}, which adds session attestation to existing TLS applications without changing applications, TLS libraries, certificates, or application protocols. \sys{} attests the protected origin of the server's current TLS key share and temporarily gates outbound traffic until verification succeeds, after which communication returns to the native TLS data path.

We make the following contributions:
\begin{enumerate}
\item \textbf{ESR analysis.} We formulate ESR as a common attack goal and develop two sufficient constructions for realizing it. By decomposing their requirements into service-identity, protected-environment, and receiving-state resources together with their reuse scope, the framework exposes the security dependencies of TLS-attestation designs.

\item \textbf{Transparent session authentication.} We design \sys{} to attest the protected origin of the server's current TLS key share while retaining independent service identity authentication. A temporary client-side gate delays sensitive traffic until verification succeeds, after which communication returns to the native TLS data path without an additional payload-encryption layer.

\item \textbf{Cross-platform realization and extensions.} We implement TLSLatch on a hardware-backed Hygon CSV CVM server, with clients on Linux, Windows, and macOS. The design further supports mutual attestation between CVMs and \batch{}, a delayed-confirmation mode that amortizes attestation across connections. We release the implementation as open source to support reproduction and further research.

\item \textbf{Detailed evaluation.} We evaluate \sys{} on a hardware-backed Hygon CSV CVM with three client platforms against native TLS and four attested-channel configurations, measuring authentication latency, transfer overhead, resource costs, mutual attestation, and batch confirmation.

%\item \textbf{Reusable implementation.} We release our implementation as open source to support reuse and further research.
\end{enumerate}

% !TeX root = ../main.tex
\section{Background}
\label{sec:background}
\subsection{TEEs and Remote Attestation}
Trusted execution environments (TEEs) use hardware-enforced isolation to protect code and data from privileged software outside the protected boundary. Process-oriented TEEs isolate selected application components, whereas confidential VMs (CVMs) protect an entire guest OS and its applications. This VM-level abstraction preserves conventional OS and application interfaces, allowing existing workloads to move into confidential execution without enclave-style application partitioning. AMD SEV-SNP, Intel TDX, Arm CCA, and Hygon CSV provide hardware support for this model~\cite{amdsevsnp,inteltdx,armcca,hygoncsv}.

%Remote attestation (RA) allows a verifier to assess a protected environment from signed evidence rooted in platform trust. Evidence may bind a fresh challenge to claims such as software measurements, platform identity, security version, and configuration; the verifier validates the evidence and evaluates these claims against reference values and an appraisal policy~\cite{rfc9334}. Attestation, however, can establish only properties captured by its evidence. An initial software measurement, for example, does not guarantee continued runtime integrity, motivating runtime measurement and code-integrity mechanisms~\cite{verismo2024}. Such mechanisms strengthen assurance about the protected environment itself, but do not by themselves establish that client data is delivered to that environment.

Remote attestation (RA) allows a verifier to assess claims about a protected environment from signed evidence rooted in platform trust. Evidence may bind a fresh challenge to software measurements, platform identity, security version, and configuration, which the verifier evaluates against an appraisal policy~\cite{rfc9334}. RA establishes the properties of the environment represented in that evidence; channel authentication must additionally bind those properties to the endpoint receiving client data.

%In this paper, We use \emph{RA} for evidence generation and appraisal, and \emph{channel authentication with attestation} for protocols that apply this evidence to a communicating endpoint.

\subsection{TLS 1.3}
\label{sec:tls-background}
In certificate-authenticated TLS 1.3, the server proves possession of the private key corresponding to its certificate through CertificateVerify, whose signature covers the handshake transcript, including the negotiated key shares. This reusable authentication credential establishes service identity and is separate from the ephemeral state of an individual session.

In a full ECDHE handshake, ClientHello and ServerHello carry the client and server key shares, written as $X=g^x$ and $Y=g^y$, from which both endpoints derive the shared secret $g^{xy}$ and, together with the handshake transcript, the handshake and traffic keys. Finished provides key confirmation for the resulting handshake state~\cite{rfc8446}. The public shares $X$ and $Y$ are visible on the wire without revealing $x$, $y$, or $g^{xy}$. Consequently, compromise of the reusable certificate key can enable service impersonation on a new connection, but does not by itself reveal the traffic keys of an honest session with fresh ECDHE state.

\begin{figure*}[t]
\centering% Native TikZ; sized for two columns without font scaling.
\begin{minipage}[t]{.49\textwidth}\centering
\textbf{(a) Direct / TLS passthrough}\par\vspace{8pt}
\begin{tikzpicture}
\node[endpoint,minimum width=1.15cm,minimum height=.85cm] (c) at (.65,0) {Client};
\node[endpoint,draw=black!45,dashed,fill=white,minimum width=1.8cm,minimum height=.85cm] (g) at (3.75,0) {Optional\\forwarder};
\node[endpoint,minimum width=1.4cm,minimum height=.85cm] (s) at (6.85,0) {Service};
\node[boundary,fit=(s),inner sep=6pt,label=above:{CVM}] {};
\draw[->,black!50] (c.east)--(g.west);
\draw[->,black!50] (g.east)--(s.west);
\draw[tls] (c.south)--++(0,-.6)-| node[pos=.25,below]{TLS to CVM} (s.south);
\draw[evidence] ($(s.south)+(.22,0)$)--++(0,-1.4)-| node[pos=.25,below]{Attestation evidence} ($(c.south)+(-.22,0)$);
\end{tikzpicture}
\end{minipage}\hfill
\begin{minipage}[t]{.49\textwidth}\centering
\textbf{(b) TLS-terminating gateway}\par\vspace{8pt}
\begin{tikzpicture}
\node[endpoint,minimum width=1.15cm,minimum height=.85cm] (c) at (.65,0) {Client};
\node[endpoint,minimum width=1.8cm,minimum height=.85cm] (g) at (3.75,0) {Gateway\\/ LB};
\node[endpoint,minimum width=1.4cm,minimum height=.85cm] (s) at (6.85,0) {Service};
\node[boundary,fit=(s),inner sep=6pt,label=above:{CVM}] {};
\draw[tls] (c.east)--node[above]{TLS}(g.west);
\draw[->,black!50] (g.east)--node[above]{Routing}(s.west);
\draw[tls] (c.south)--++(0,-.6)-| node[pos=.25,below]{Inner ciphertext to CVM} (s.south);
\draw[evidence] ($(s.south)+(.22,0)$)--++(0,-1.4)-| node[pos=.25,below]{Attestation evidence} ($(c.south)+(-.22,0)$);
\end{tikzpicture}
\end{minipage}
\caption{The protected worker remains the request-decryption endpoint. (a) TLS terminates inside the CVM. (b) A gateway terminates outer TLS, while an inner encrypted payload reaches the CVM.}
\Description{Two deployment models compare direct or TLS-passthrough access with a TLS-terminating gateway. Dashed boundaries mark the CVM; dashed return arrows carry attestation evidence.}
\label{fig:cloud-deployments}
\end{figure*}
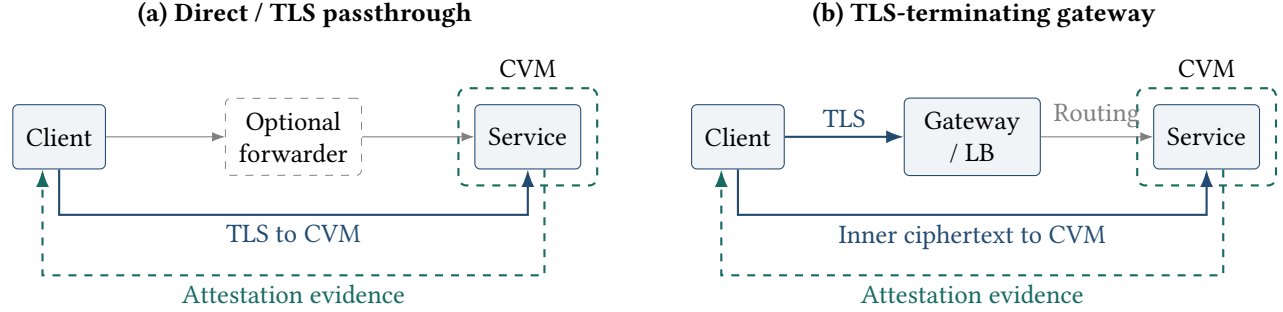

\subsection{Confidential Cloud Services and Deployment Models}
\label{sec:cloud-background}
Confidential cloud services aim to protect sensitive request processing through different trusted-computing architectures. Confidential AI inference is a prominent example: Alibaba Cloud's CMaaS provides confidential inference in a CVM setting and publishes its client SDK~\cite{aliyuncmaas,cmaassdk}. Related services pursue similar privacy goals with different architectures~\cite{applepcc,googleprivate,metaprivate}. Apple PCC, for example, relies on dedicated Apple silicon, secure and measured boot, hardware attestation, and software transparency rather than tenant CVMs~\cite{applehardwaretrust}. Despite these architectural differences, sensitive requests must ultimately reach an authenticated protected receiver. Clients therefore need platform-backed evidence about the intended receiving environment in addition to conventional service authentication.

Figure~\ref{fig:cloud-deployments} shows two deployment structures. In a direct or TLS-passthrough deployment, the client’s TLS connection terminates in the protected worker; intermediate forwarders do not decrypt it. In a TLS-terminating gateway deployment, the gateway terminates the outer TLS connection. To keep application plaintext hidden from the gateway, the request is carried in an inner encrypted channel that terminates in the protected worker. CMaaS follows this pattern through application-key establishment and payload encryption~\cite{cmaassdk}. We model the former as a single protected endpoint and the latter as layered outer and inner endpoints.

%TLS termination at a gateway can preserve request confidentiality when the request has another protection layer whose decryption endpoint remains inside the protected worker. Before sending sensitive data, the client authenticates that endpoint's key material and establishes an application-level encryption key, or encapsulates a request key to it. The gateway removes the outer TLS layer but still handles application ciphertext. CMaaS follows this pattern through application key establishment and payload encryption~\cite{cmaassdk}. The two encryption layers serve different endpoints. Since a transparent gateway can be inserted without changing the endpoint authentication problem, we use the direct model below and return to gateway termination where it changes the credentials or encryption layers an attack must overcome.

% !TeX root = ../main.tex
\section{Endpoint-Substitution Relays and Compromise Robustness}
\label{sec:esr}

An \emph{Endpoint-Substitution Relay} (ESR) occurs when a client releases a secret request to an accepted receiver whose plaintext is available to the adversary, while the genuine protected service is used to preserve the expected computation. ESR abstracts this outcome across attested-channel protocols by characterizing the resources sufficient to realize it under a fixed protocol, deployment, and verification policy.

Two alternative sufficient constructions realize ESR under different capabilities. The \emph{attested-environment path} uses a live protected environment accepted by the client's policy, while the \emph{receiving-state copying path} reproduces an attestation-endorsed receiving capability elsewhere. This decomposition exposes the service-identity, protected-environment, and receiving-state resources on which different designs depend, together with the scope over which acquired state remains usable. Prior work identifies concrete relay and diversion failures in attested TLS constructions~\cite{weinhold2025tlsra,sardar2026identity}; ESR places these cases under a common outcome and makes their resource dependencies explicit.

Table~\ref{tab:notation} summarizes the notation used throughout the analysis.
% !TeX root = ../main.tex
\begin{table}[t]
	\caption{Notation used in the ESR analysis.}
	\label{tab:notation}
	\centering
	\setlength{\tabcolsep}{4pt}
	\renewcommand{\arraystretch}{1.05}
	\begin{tabularx}{\columnwidth}{@{}>{\centering\arraybackslash}p{1.45cm}>{\raggedright\arraybackslash}X@{}}
		\toprule
		Symbol & Meaning \\
		\midrule
		$\Pi,\pi$ & Fixed configuration and target context. \\
		$\sigma_\Pi,B_\Pi$ & Protected receiving state used by the copying path; acquisition of that state. \\
		$pk_b,sk_b$ & Service-authentication public/private key in the shared-key example. \\
 $\begin{gathered}X=g^x\\Y=g^y\end{gathered}$ & Client/server ephemeral TLS key shares. \\
		$T_0$ & Freely generated self-signed identity. \\
		$T_E,T_S$ & Enterprise / service-restricted authentication capability. \\
		$R_0$ & Ordinary equipment with no attestation qualification. \\
		$R_V,R_P$ & Vendor-accepted protected environment / added provider authorization. \\
		\bottomrule
	\end{tabularx}
\end{table}

\subsection{Participants and Adversarial Capabilities}
\label{sec:threat}

The participants are an honest client $C$, a production service $S$, and an adversary $A$. The adversary can redirect the client’s traffic to a substitute endpoint and invoke $S$ through its ordinary interfaces. These are baseline capabilities of the ESR model; the service credentials, protected environments, and receiving state available to $A$ are modeled separately below. These capabilities may arise from privileged infrastructure access or an external compromise, without granting access to the production protected environment or its secrets.

%Our analysis is capability-based. For each construction, we state which service credentials, protected environments, or receiving state are available to $A$; acquiring one resource does not imply access to unrelated resources. Unless explicitly stated otherwise, the production service and client verification execute correctly, and attestation evidence accepted by the client's policy cannot be forged. A configuration $\Pi$ fixes the protocol, deployment, and client verification policy under analysis.

The analysis is capability-based. Each construction specifies the service credentials, protected environments, and receiving state available to $A$. These resources are modeled independently unless a construction explicitly links them. The production service and client verification execute correctly, and attestation evidence accepted by the client’s policy is unforgeable. A configuration $\Pi$ fixes the protocol, deployment, and verification policy under analysis.

\subsection{ESR Goal and Resource Model}
\label{sec:esr-definition}

To characterize ESR independently of a particular protocol, we first isolate the accepted-request disclosure event that every construction must realize.

\begin{definition}[Accepted-request disclosure]
	\label{def:esr}
	For configuration $\Pi$, target context $\pi$, and secret request $p$,
	\[
	\mathsf{Win}_{\Pi}
	= \mathsf{Accept}_{C}(\pi)
	\land \mathsf{Send}_{C}(p)
	\land \mathsf{Recover}_{A}(p).
	\]
	$\mathsf{Accept}_{C}(\pi)$ means that the receiver satisfies all identity, environment, and binding checks required by $\Pi$ for the target context. $\mathsf{Recover}_{A}(p)$ means that the adversary
	obtains the plaintext request $p$ outside the service's authorized
	output; if multiple protection layers cover $p$, recovery requires
	defeating all of them.
\end{definition}

An ESR is an accepted-request disclosure in which the accepted receiver
is a substitute endpoint and the recovered request is relayed to $S$ to
preserve the expected behavior. Leakage from an already established
honest session is a separate confidentiality failure rather than endpoint
substitution.

The two constructions are characterized by three resource classes. $T_\Pi$ denotes the service-authentication capability required to satisfy the identity checks of $\Pi$; for certificate-authenticated TLS, this includes the corresponding private signing capability. $R_\Pi$ denotes access to the live protected environment required by the attestation policy. $B_\Pi$ denotes acquisition of the protected receiving state $\sigma_\Pi$ defined in Section~\ref{sec:copying-path}.

We distinguish three identity cases. $T_0$ denotes a freely generated self-signed identity when no independently provisioned service identity is required; a specifically pinned certificate is therefore not $T_0$. $T_E$ denotes an enterprise credential accepted for the target service, while $T_S$ additionally restricts that credential to a specific service, workload, or instance.

For environment resources, $R_0$ denotes ordinary equipment with no attestation qualification. $R_V$ denotes access to a protected environment whose evidence is accepted through the platform vendor's attestation trust chain. $R_P$ adds deployment-specific authorization by the service or cloud provider~\cite{intelquote,applehardwaretrust}. These labels capture distinct policy requirements rather than a security ordering; whether provider authorization creates an independent acquisition barrier depends on how it is provisioned.

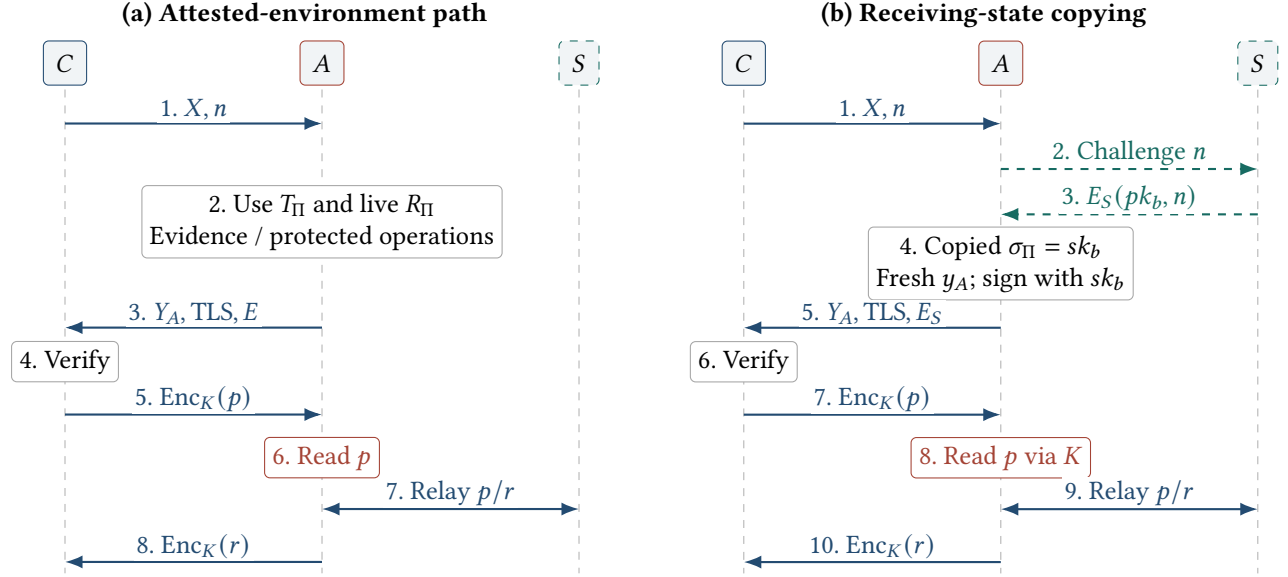
\begin{figure*}[t]
	\centering
	\begin{minipage}[t]{.49\textwidth}\centering
\textbf{(a) Attested-environment path}\par\vspace{4pt}
\begin{tikzpicture}[msg/.style={fill=white,inner sep=1pt,align=center},event/.style={draw=black!35,fill=white,rounded corners=2pt,inner sep=3pt,align=center}]
\node[endpoint] at (.6,0) {$C$};\node[endpoint,draw=brick] at (4,0) {$A$};\node[endpoint,draw=teal,dashed] at (7.4,0) {$S$};
\foreach \x in {.6,4,7.4}{\draw[black!30,dashed] (\x,-.35)--(\x,-6.75);}
\draw[tls] (.6,-.8)--node[msg,above]{1. $X,n$}(4,-.8);
\node[event] at (4,-2.1){2. Use $T_\Pi$ and live $R_\Pi$\\Evidence / protected operations};
\draw[tls] (4,-3.5)--node[msg,above]{3. $Y_A,\mathrm{TLS},E$}(.6,-3.5);
\node[event] at (.6,-3.95){4. Verify};
\draw[tls] (.6,-4.65)--node[msg,above]{5. $\operatorname{Enc}_K(p)$}(4,-4.65);
\node[event,draw=brick,text=brick] at (4,-5.22){6. Read $p$};
\draw[<->,navy,thick] (4,-5.9)--node[msg,above]{7. Relay $p/r$}(7.4,-5.9);
\draw[tls] (4,-6.6)--node[msg,above]{8. $\operatorname{Enc}_K(r)$}(.6,-6.6);
\end{tikzpicture}
\end{minipage}\hfill
\begin{minipage}[t]{.49\textwidth}\centering
\textbf{(b) Receiving-state copying}\par\vspace{4pt}
\begin{tikzpicture}[msg/.style={fill=white,inner sep=1pt,align=center},event/.style={draw=black!35,fill=white,rounded corners=2pt,inner sep=3pt,align=center}]
\node[endpoint] at (.6,0) {$C$};\node[endpoint,draw=brick] at (4,0) {$A$};\node[endpoint,draw=teal,dashed] at (7.4,0) {$S$};
\foreach \x in {.6,4,7.4}{\draw[black!30,dashed] (\x,-.35)--(\x,-6.75);}
\draw[tls] (.6,-.8)--node[msg,above]{1. $X,n$}(4,-.8);
\draw[evidence] (4,-1.4)--node[msg,above]{2. Challenge $n$}(7.4,-1.4);
\draw[evidence] (7.4,-2.0)--node[msg,above]{3. $E_S(pk_b,n)$}(4,-2.0);
\node[event] at (4,-2.65){4. Copied $\sigma_\Pi=sk_b$\\Fresh $y_A$; sign with $sk_b$};
\draw[tls] (4,-3.5)--node[msg,above]{5. $Y_A,\mathrm{TLS},E_S$}(.6,-3.5);
\node[event] at (.6,-3.95){6. Verify};
\draw[tls] (.6,-4.65)--node[msg,above]{7. $\operatorname{Enc}_K(p)$}(4,-4.65);
\node[event,draw=brick,text=brick] at (4,-5.22){8. Read $p$ via $K$};
\draw[<->,navy,thick] (4,-5.9)--node[msg,above]{9. Relay $p/r$}(7.4,-5.9);
\draw[tls] (4,-6.6)--node[msg,above]{10. $\operatorname{Enc}_K(r)$}(.6,-6.6);
\end{tikzpicture}
\end{minipage}
	\caption{Two alternative sufficient ESR constructions with the same read-and-relay outcome. (a) The attested-environment path uses a live protected environment accepted by the client's policy; the client-facing receiver may be colocated with it or composed with it remotely. (b) The receiving-state copying example uses the endorsed $sk_b$, reuses applicable evidence for $pk_b$, and authenticates an attacker-chosen key exchange. In both, $K$ denotes the client--attacker traffic keys; the genuine service is invoked over a separate protected connection. Selected logical messages are shown, and $n$ is used when the scheme requires a challenge.}
	\Description{Two side-by-side interaction diagrams compare using a live accepted protected environment with reproducing an endorsed receiving capability from copied state, followed by decryption and relay.}
	\label{fig:esr}
\end{figure*}

\subsection{Attested-Environment Path}
\label{sec:device-path}

The attested-environment path uses a live protected environment accepted by the client’s attestation policy. The substitute receiver may terminate the client-facing interaction in that environment, or combine an ordinary front end with a remotely accessible protected environment that performs the required attested operations. The path applies when the resulting composition satisfies the checks of $\Pi$ and allows the adversary to recover the protected request.

\begin{proposition}[Sufficiency of the attested-environment path]
	\label{prop:device}
	Under these acceptance and recoverability conditions,
	\begin{equation}
		T_\Pi \land R_\Pi
		\quad\Longrightarrow\quad
		\text{successful ESR}.
		\label{eq:device}
	\end{equation}
\end{proposition}

With \(T_\Pi\), the substitute receiver satisfies the service-identity checks of \(\Pi\); with \(R_\Pi\), it obtains the required attestation evidence or protected protocol operations. After the client accepts and releases the request, the adversary recovers \(p\), relays it to \(S\) over a separate connection, and returns the genuine result. This path uses the live protected environment directly and requires no copied receiving state.

\subsection{Receiving-State Copying Path}
\label{sec:copying-path}
\label{sec:copy-path}
\label{sec:binding-analysis}

The receiving-state copying path reproduces an attestation-endorsed receiving operation outside a policy-accepted protected environment while retaining applicable evidence. Let \(E\) denote the evidence and \(b\) its authenticated binding data. We use \(\sigma_\Pi\) for the protected server-side source state underlying the evidence-bound receiving capability. The boundary of \(\sigma_\Pi\) excludes secrets whose disclosure only exposes an already established honest session.

\begin{observation}[Evidence and receiving capability]
	\label{obs:binding}
	The channel protection contributed by attestation is mediated by the receiving capability authenticated by its evidence. For the copying path, \(\sigma_\Pi\) denotes the protected server-side source state underlying that capability. If \(\sigma_\Pi\) can be reproduced elsewhere while the evidence remains applicable, the corresponding receiving role can be reproduced there.
\end{observation}

$B_\Pi$ denotes acquisition of \(\sigma_\Pi\). The copying path requires that \(\sigma_\Pi\), together with available protocol material and attacker-chosen randomness, can reproduce the endorsed receiving operation and recover \(p\), and that the corresponding evidence remains applicable to the new interaction. Any additional state, protected operation, or authorization required by \(\Pi\) is part of these premises.

\begin{proposition}[Sufficiency of receiving-state copying]
	\label{prop:copy}
	Under the reproducibility and evidence-applicability premises above,
	\begin{equation}
		T_\Pi \land R_0 \land B_\Pi
		\quad\Longrightarrow\quad
		\text{successful ESR}.
		\label{eq:copy}
	\end{equation}
\end{proposition}

With \(T_\Pi\), the substitute receiver satisfies the service-identity checks of \(\Pi\). With \(R_0\) and \(B_\Pi\), the adversary instantiates the reproduced receiving operation on ordinary equipment and presents applicable evidence. After the client accepts and releases the request, the substitute receiver recovers \(p\), relays it to \(S\), and returns the genuine result.

Figure~\ref{fig:esr}(b) instantiates the copying path for shared authentication-key binding~\cite{weinhold2025tlsra}. The reusable private key \(sk_b\) serves as the receiving state \(\sigma_\Pi\), while applicable evidence endorses its public counterpart \(pk_b\). After acquiring \(sk_b\), the attacker chooses a fresh ephemeral key \(y_A\), authenticates that exchange with \(sk_b\), and obtains applicable evidence for \(pk_b\), including a fresh challenge when required. Because the evidence endorses the reusable authentication key rather than the receiving state of the new exchange, the attacker derives the client-facing traffic keys from \(y_A\), recovers \(p\), and relays it to the production service.

For this shared-key configuration, the same reusable private state supplies both service authentication and the attestation-endorsed receiving capability, so \(T_\Pi \Rightarrow B_\Pi\). Under Proposition~\ref{prop:copy}, the copying path therefore reduces from \(T_\Pi \land R_0 \land B_\Pi\) to \(T_\Pi \land R_0\). When the endorsed receiving state is independent of the reusable service credential, \(T_\Pi\) does not supply \(B_\Pi\). The two paths are sufficient constructions for ESR rather than an exhaustive characterization of confidentiality failures.

\subsection{Mapping Existing Designs}
\label{sec:matrix}

%We now map representative attested-channel designs to the resource model. Table~\ref{tab:comparison} records the service identity $T_\Pi$, the protected-environment requirement $R_\Pi$, and the protected receiving state $\sigma_\Pi$ relevant to the copying path. Some rows describe general constructions that admit several deployment choices, whereas others describe concrete systems with a fixed identity and attestation policy; circles therefore denote compatible deployment alternatives rather than resources an adversary may choose during an attack.

Table \ref{tab:comparison} maps representative attested-channel designs to the ESR resource model. For each scheme or configuration, it records the required service-identity resource \(T_\Pi\), protected-environment resource \(R_\Pi\), and protected receiving state \(\sigma_\Pi\). The table then separates two properties relevant to the copying path: whether the service credential supplies that receiving state (\(T_\Pi \Rightarrow B_\Pi\)), and the scope over which the attestation-to-receiving-capability binding remains applicable. The discussion below groups related mechanisms for exposition rather than partitioning the design space into disjoint classes.

We classify this scope as long-lived, periodic, or per-session. Long-lived bindings remain reusable with persistent credential state; periodic bindings are renewed by deployment policy; and per-session bindings are tied to the current protocol session by construction. A periodic deployment may refresh as often as once per session, but such freshness is policy-driven rather than intrinsic to the construction.

%The table also records whether compromise of the required service credential supplies the receiving state ($T_\Pi\Rightarrow B_\Pi$), and the freshness of the attestation-to-receiving-capability binding. \emph{Long-lived} denotes a binding reusable with persistent credential state, \emph{periodic} denotes policy-controlled renewal, and \emph{per-session} denotes a binding intrinsically tied to the current protocol session. Periodic renewal may be configured as frequently as every session; the distinction is whether such freshness is required by the construction itself.

% !TeX root = ../main.tex
\begin{table*}[tp]
	\caption{Resources and binding properties for the two sufficient ESR
		paths. The attested-environment path uses $T_\Pi,R_\Pi$
		(Equation~\eqref{eq:device}); the receiving-state copying path uses
		$T_\Pi,R_0,B_\Pi$ (Equation~\eqref{eq:copy}).}
	\label{tab:comparison}
	\centering
	\setlength{\tabcolsep}{3pt}
	\renewcommand{\arraystretch}{1.12}
	
	\begin{tabularx}{\textwidth}{
			@{}
			>{\raggedright\arraybackslash}p{4.25cm}
			*{6}{>{\centering\arraybackslash}p{.38cm}}
			>{\raggedright\arraybackslash}X
			>{\centering\arraybackslash}p{1.25cm}
			>{\centering\arraybackslash}p{1.95cm}
			@{}}
		\toprule
		Scheme / configuration
		& \multicolumn{3}{c}{Identity $T_\Pi$}
		& \multicolumn{3}{c}{Environment}
		& \multicolumn{3}{c}{Receiving-state acquisition $B_\Pi$} \\
		\cmidrule(lr){2-4}
		\cmidrule(lr){5-7}
		\cmidrule(lr){8-10}
		& $T_0$ & $T_E$ & $T_S$
		& $R_0$ & $R_V$ & $R_P$
		& Required state $\sigma_\Pi$
		& $T_\Pi\!\Rightarrow\!B_\Pi$
		& Freshness \\
		\midrule
		
		RA-TLS~\cite{knauth2019ratls}
		& \dc & & & & \op & \op
		& TLS authentication key
		& -- & Periodic \\
		
		Platform Certificate~\cite{goldman2006linking}
		& & \op & \op & & \op & \op
		& Shared service key
		& Yes & Long-lived \\
		
		RATLS~\cite{walther2022ratls}
		& & \op & \op & & \op & \op
		& Shared service key
		& Yes & Long-lived \\
		
		RA-TLS + CA~\cite{knauth2019ratls}
		& & \op & \op & & \op & \op
		& CA- and RA-endorsed key
		& Yes & Long-lived \\
		
		Trusted Channels~\cite{armknecht2008trusted}
		& & \op & \op & & \op & \op
		& Shared service key
		& Yes & Long-lived \\
		
		TC4SE, initialized TLS~\cite{hamidy2023tc4se}
		& & & \dc & \dc & & 
		& Pinned authentication key
		& Yes & Long-lived \\
		
		Attested CSR, later TLS version~\cite{attestedcsr14}
		& & \op & \op & \dc & &
		& Issued authentication key
		& Yes & Long-lived \\
		
		TLS+RA~\cite{weinhold2025tlsra}
		& & \op & \op & & \op & \op
		& Current ephemeral private key
		& No & \mbox{Per-session} \\
		
		HTTPA/1 + HTTPS~\cite{king2023httpa}
		& & \op & \op & & \op & \op
		& Attested unwrapping key
		& No & Periodic \\
		
		TNG-single~\cite{tngsoftware}
		& \dc & & & & \op & \op
		& Tunnel authentication key
		& -- & Periodic \\
		
		TNG-nested~\cite{tngsoftware}
		& & \op & \op & & \op & \op
		& Outer tunnel authentication key
		& No & Periodic \\
		
		CMaaS~\cite{cmaassdk}
		& & \dc & & & \dc &
		& Node negotiation key
		& No & Periodic \\
		
		Apple PCC~\cite{applerequesthandling,applepccsoftware}
		& & \dc & & & & \dc
		& REK decapsulation key
		& No & Periodic \\
		
		\textbf{\sys{}}
		& & \op & \op & & \op & \op
		& \textbf{Current ephemeral private key}
		& \textbf{No} & \textbf{\mbox{Per-session}} \\
		
		\bottomrule
	\end{tabularx}
	
	\par\vspace{3pt}
	\begin{minipage}{\textwidth}
		\textbf{Resources.}
		$\dc$: stated configuration;
		$\circ$: compatible deployment alternative (one fixed choice per group);
		blank: not selected.
		Some rows describe general constructions that admit several deployment
		choices, whereas others represent a concrete system configuration.
		
		\textbf{Credential-state overlap.}
		``Yes'' means that acquiring the service-authentication capability required by \(T_\Pi\) also supplies the specific receiving state \(\sigma_\Pi\), i.e., \(T_\Pi \Rightarrow B_\Pi\);
		``No'' means that the two resources remain distinct;
		``--'' denotes \(T_0\) configurations in which the service identity is generated within the attested environment rather than provisioned as an independent credential; the \(T_\Pi \Rightarrow B_\Pi\) overlap is therefore not separately defined.
		
		\textbf{Freshness.}
		$Long-lived$: reusable with persistent state; $Periodic$: renewed by deployment policy; $Per-session$: intrinsically bound to the current session.
	\end{minipage}
\end{table*}

\paragraph{Attestation-bound authentication state.}
A common construction binds platform evidence to reusable TLS authentication state. Typical RA-TLS binds an enclave-generated TLS key to SGX evidence and typically uses it through a self-signed certificate~\cite{knauth2019ratls}. Platform Certificate, RATLS, Trusted Channels, and TC4SE similarly bind platform trust to reusable authentication state, while Attested CSR performs the binding during certificate enrollment~\cite{goldman2006linking,walther2022ratls,armknecht2008trusted,hamidy2023tc4se,attestedcsr14}. In self-signed RA-TLS, the TLS identity is generated inside the attested environment rather than provisioned as an independent enterprise credential. Where the same reusable private state supplies both service authentication and the evidence-bound receiving capability, \(T_\Pi \Rightarrow B_\Pi\). RA-TLS with a CA adds conventional PKI identity while retaining this overlap.

%For the stated configurations in which the same reusable private state supplies both service authentication and the evidence-bound receiving capability, compromise of the service credential also supplies $B_\Pi$, yielding $T_\Pi\Rightarrow B_\Pi$. A self-signed $T_0$ construction avoids dependence on a separately provisioned service credential, but correspondingly provides no independent enterprise or service identity. RA-TLS with a CA can retain the same attested leaf key while adding conventional PKI identity; in that configuration the credential and receiving state still coincide.

TNG is a mainstream open-source transparent gateway for confidential computing that applies this attested-TLS pattern at the network layer~\cite{tngsoftware}. In its RA-TLS mode, attestation binds the tunnel certificate key, while the gateways transparently carry application traffic over the resulting encrypted tunnel~\cite{tngsoftware}. TNG-single uses this tunnel as the application protection layer. TNG-nested instead carries the application's original TLS connection inside the tunnel, so the tunnel authentication state and application service credential remain distinct, and the inner TLS connection remains a separate protection layer for \(\mathsf{Recover}_A(p)\).

\paragraph{Application-layer attested channels.}
HTTPA preserves HTTPS and establishes an additional attested channel above it~\cite{king2023httpa}. Its evidence binds a TEE-generated public key whose private key unwraps fresh client-provided pre-session material, from which trusted-session keys are derived. The protected receiving state \(\sigma_\Pi\) is therefore the attested unwrapping private key. Because the attestation result may be cached and reused, the binding is periodic. The outer HTTPS endpoint alone does not satisfy \(\mathsf{Recover}_A(p)\) while the HTTPA protection remains intact.

\paragraph{Deployed confidential-cloud mechanisms.}
CMaaS separates gateway-facing HTTPS from an independently protected channel to the confidential inference node~\cite{cmaassdk}. Attestation binds a node public key rooted in protected node state, under which fresh session material is established for individual interactions. We therefore take \(\sigma_\Pi\) to be the attested node root private key and classify the binding as periodic. Apple PCC uses a different hardware and service architecture but exhibits a similar resource-level structure: request secrets are released to attested node public keys backed by protected request-decryption state~\cite{applerequesthandling,applepccsoftware}. Its receiving state is likewise periodically refreshed rather than intrinsically tied to one request.

\paragraph{Current-session TLS bindings.}
TLS+RA retains conventional certificate-based service authentication while binding attestation to the current TLS handshake through the DHE-derived shared secret and handshake context~\cite{weinhold2025tlsra}. For the copying path, \(\sigma_\Pi\) is the server's current ECDHE private key \(y\), from which the shared secret is derived. The reusable service credential is independent of \(y\), so \(T_\Pi \nRightarrow B_\Pi\), and the binding is per-session.

These designs differ in protocol structure and integration point, but the ESR model evaluates them through the same questions: which protected receiving state must be acquired, whether service authentication also supplies that state, and how long the corresponding attestation binding remains applicable. For layered constructions, \(\mathsf{Recover}_A(p)\) additionally requires overcoming every protection layer that still covers the request.

\subsection{Design Implications}
\label{sec:design-implications}
\label{sec:integration}

The ESR analysis yields several design implications. First, hardware attestation alone does not establish service authorization; PCC-style services retain security dependencies in deployment authorization and infrastructure control. Second, service authentication and attestation should provide independent barriers, avoiding constructions in which compromise of a reusable service credential also supplies the protected receiving state (\(T_\Pi \Rightarrow B_\Pi\)). Third, the reuse scope of attestation-endorsed receiving state should be bounded, although the appropriate lifetime need not coincide with a single session. Finally, attestation should authenticate the receiving capability that actually protects the request, and the binding should be verified before sensitive data is released.

Section 4 applies these implications to TLSLatch, which binds attestation to the receiving state of the current TLS connection while preserving independent service authentication and verification before release.

%The analysis suggests three requirements for robust session authentication. First, attestation should authenticate the receiving capability that actually protects the client's request, rather than only a reusable credential adjacent to that channel. Second, compromise of a reusable service credential should not by itself supply the attestation-endorsed receiving state; in the resource model, the design should avoid $T_\Pi\Rightarrow B_\Pi$. Third, the binding must be verified before sensitive data is released to the receiver. Section~\ref{sec:design} presents TLSLatch, which realizes these requirements by binding attestation to the receiving state of the current TLS connection while retaining independent service authentication.

% !TeX root = ../main.tex
\section{TLSLatch Design}
\label{sec:design}
\label{sec:security}

%Section~\ref{sec:design-implications} identifies three requirements for
%robust session authentication: attestation should bind to the receiving
%capability of the current interaction, service identity should remain
%independent of the attested receiving state, and sensitive data should
%not be released before the binding is verified. TLSLatch realizes these
%requirements around an otherwise unmodified full-ECDHE TLS~1.3
%connection.
%
%Representative designs place attestation at different integration
%points. Some integrate it into certificates, the TLS stack, or the
%application protocol; transparent proxies can avoid application changes
%but may retain an additional encrypted payload path.
%Table~\ref{tab:engineering} summarizes these choices. TLSLatch instead
%uses two privileged components outside the TLS library: a server-side
%observer establishes the protected provenance of the current TLS key
%share, while a client-side gate delays sensitive traffic until that
%provenance is verified. After verification, application traffic follows
%the original TLS data path.
Translating session binding into a deployable system requires more than satisfying the security conditions of Section 3. Confidential cloud services are typically built around existing applications, TLS stacks, certificates, and network infrastructure, so changes to these components introduce development, deployment, and maintenance costs. Additional protection layers can also remain on the data path and impose per-request processing overhead. TLSLatch therefore adds session attestation around an existing TLS connection rather than replacing or encapsulating it. It preserves the application, TLS stack, certificates, and application protocol, and returns traffic to the native TLS path after authentication. This model also fits TLS-passthrough deployments, where intermediate load balancers can continue forwarding the original encrypted connection without participating in attestation. Table 3 compares these deployment properties with representative attested-channel designs.

TLSLatch realizes this design with two components outside the TLS library. A server-side observer establishes the protected provenance of the current TLS key share, while a client-side gate delays sensitive traffic until verification succeeds. Once authorized, the gate is removed from the payload path and the original TLS connection proceeds normally.

% !TeX root = ../main.tex
\begin{table}[t]
	\caption{Representative integration points and additional payload
		encryption.}
	\label{tab:engineering}
	\centering
	\setlength{\tabcolsep}{3pt}
	\renewcommand{\arraystretch}{1.10}
	
	\begin{tabularx}{\columnwidth}{
			@{}
			>{\hsize=1.52\hsize\linewidth=\hsize\raggedright\arraybackslash}X
			>{\hsize=.60\hsize\linewidth=\hsize\centering\arraybackslash}X
			>{\hsize=1.04\hsize\linewidth=\hsize\centering\arraybackslash}X
			>{\hsize=.84\hsize\linewidth=\hsize\centering\arraybackslash}X
			@{}}
		\toprule
		Scheme
		& \shortstack{App\\changes}
		& \shortstack{Integration\\point}
		& \shortstack{Extra\\encryption} \\
		\midrule
		
		RA-TLS~\cite{knauth2019ratls}
		& Yes & Certificate & No \\
		
		TLS+RA~\cite{weinhold2025tlsra}
		& Yes & TLS stack & No \\
		
		HTTPA~\cite{king2023httpa}
		& Yes & Application & Yes \\
		
		TNG single~\cite{tngsoftware}
		& No & Proxy & No \\
		
		TNG nested~\cite{tngsoftware}
		& No & Proxy & Yes \\
		
		CMaaS SDK~\cite{cmaassdk}
		& Yes & SDK & Yes \\
		
		CMaaS proxy~\cite{cmaassdk}
		& No & Proxy & Yes \\
		
		\textbf{\sys{}}
		& \textbf{No} & \textbf{Agent} & \textbf{No} \\
		
%		\sys{}-Batch
%		& No & Agent & No \\
		
		\bottomrule
	\end{tabularx}
	
	\par\vspace{3pt}
	\raggedright
	Application changes exclude deployment configuration such as connection endpoints and trust policies.
\end{table}

\subsection{Bind Attestation to the Current TLS Receiver}
\label{sec:binding}

In a full TLS 1.3 ECDHE handshake, the server TLS endpoint generates the ephemeral
share $Y=g^y$ and sends it in ServerHello. TLSLatch binds the evidence to this share by hashing a role-tagged encoding of the negotiated group and $Y$:
\begin{equation}
	b=\Hash(\operatorname{Encode}(\text{server-role tag},\text{group},Y)).
	\label{eq:server-binding}
\end{equation}
The evidence carries $b$ together with the platform claims used by the
client's attestation policy. The client independently computes the same
digest from the ServerHello of its own connection and verifies
the evidence and binding. Evidence for another share therefore cannot
authenticate this connection. Because \(Y\) is public handshake data, TLSLatch can construct this binding from observable protocol state without extracting ephemeral secret state.

Binding evidence to \(Y\) requires establishing the protected provenance of that share. TLSLatch therefore attests only a server share associated with an active local TLS flow. The observer establishes which local flow emitted \(Y\), while the approved guest kernel and TLS implementation provide the key-generation and key-use semantics for that flow inside the CVM. Together, these guarantees bind \(Y\) to the protected TLS endpoint that participates in the exchange. Section 5 describes how the prototype realizes this flow association without modifying the TLS library.

%These requirements separate local observation from TLS semantics.
%Observation establishes which TLS flow emitted the public share, while
%the approved guest kernel and TLS implementation provide the
%key-generation and key-use semantics. Mere packet transit through the
%protected environment is not sufficient. Section~\ref{sec:implementation}
%describes how the prototype associates observed shares with active local
%flows without modifying the TLS library.

\subsection{Verify Before Release}
\label{sec:gate}

The client installs a per-connection gate before releasing protected traffic. ClientHello is allowed to proceed, while attestation runs in parallel with the unmodified TLS handshake. The server’s normal handshake flight reaches the client TLS stack, allowing certificate validation and TLS key derivation to proceed while attestation is in progress. Figure~\ref{fig:latch-flow} shows this overlap.

\begin{figure*}[t]
	\centering
	\begin{tikzpicture}[msg/.style={fill=white,inner sep=2pt,align=center}]
\node[endpoint] at (0,0) {Client TLS};
\node[endpoint] at (7,0) {Client agent};
\node[endpoint,draw=teal] at (14,0) {CVM TLS + agent};
\foreach \x in {0,7,14}{\draw[black!30,dashed](\x,-.4)--(\x,-4.9);}
\draw[tls](0,-.85)--node[msg,above]{1. ClientHello: observe lookup share; allow transmission}(14,-.85);
\draw[evidence](7,-1.55)--node[msg,above]{2. Query local connection record}(14,-1.55);
\draw[tls](14,-2.25)--node[msg,above]{3. ServerHello + normal TLS server flight}(0,-2.25);
\node[msg,draw=brick,text=brick] at (2.5,-3.05){Hold encrypted outbound records};
\draw[evidence](14,-3.05)--node[msg,above]{4. Evidence for local server share}(7,-3.05);
\draw[evidence](7,-3.9)--node[msg,above]{5. Verify; authorize live flow; release}(0,-3.9);
\draw[tls](0,-4.65)--node[msg,above]{6. Original TLS finishes; payload bypasses agent processing}(14,-4.65);
\end{tikzpicture}
	\caption{TLSLatch changes release timing rather than TLS messages or
		payload encryption. ClientHello proceeds normally while the client gate
		holds later outbound records. In parallel, the server attests the
		protected provenance of its registered TLS share. After verification,
		the held records are released and subsequent application traffic follows
		the native TLS path.}
	\Description{The original TLS connection passes through a temporary
		client gate to the protected server. A parallel attestation exchange
		verifies the server share observed on that connection. After successful
		verification, the gate releases the connection and subsequent payload
		traffic follows the original TLS path.}
	\label{fig:latch-flow}
\end{figure*}
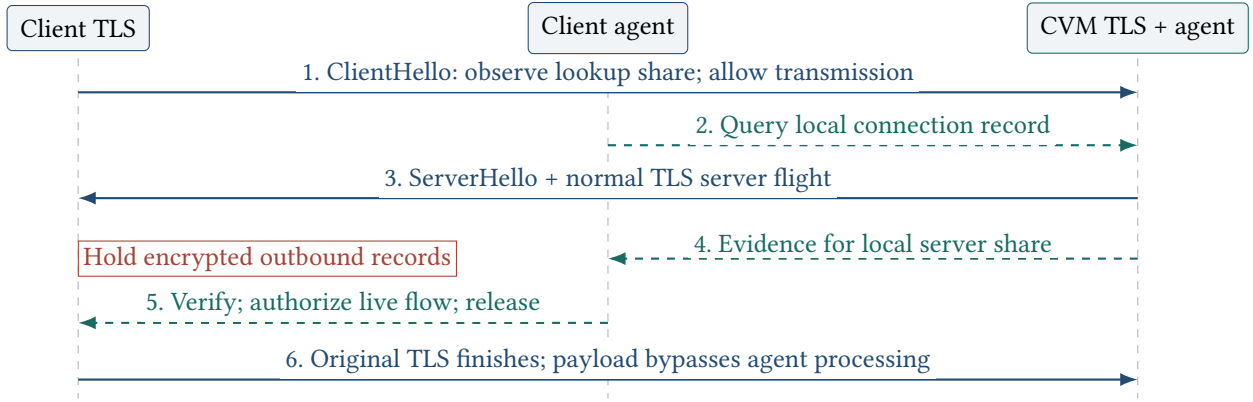

%The gate holds outbound TLS records after ClientHello, including client
%Finished and any application records queued behind it. The gate treats outbound TLS records as opaque and releases them unchanged after verification. The server emits its
%key share and server handshake flight before receiving client Finished,
%so evidence for $Y$ can be generated while the client gate holds its
%outbound flight. After authorization, the temporary control is removed
%and the agent no longer participates in application-data processing.
%The implementation mechanisms used to retain, authorize, and release a
%flow are described in Section~\ref{sec:implementation}.

The gate then holds subsequent outbound TLS records, including client Finished and any queued application data. Because the server can send its handshake flight before receiving client Finished, evidence for \(Y\) can be generated and verified while these records are held. If attestation succeeds, the gate marks the connection as authorized and releases the queued records; otherwise, it drops them. The TLS stack independently performs its normal certificate, transcript, and Finished validation. After authorization, subsequent traffic bypasses the gate and continues on the native TLS path.

\begingroup
\setlength{\emergencystretch}{3em}
\subsection{Security Argument}
\label{sec:endpoint-guarantee}
We formalize the service-credential-only case of the receiving-state copying path. The configuration uses fresh, full TLS~1.3 ECDHE handshakes. The adversary controls the network, holds the accepted service signing credential, and may invoke the genuine service through its ordinary interface. The client and the policy-accepted protected endpoint execute correctly. Request recovery means disclosure outside the service's authorized output, as in Definition~\ref{def:esr}.

\subsubsection{Execution Model and Assumptions}
For a client session $s$, let $X_s=g^{x_s}$ and $Y_s=g^{y_s}$ denote its negotiated client and server shares, and let
\[
 b_s=H(\operatorname{Encode}(\mathsf{server},\mathsf{group}_s,Y_s)).
\]
The encoding is injective and role-separated. Write $\mathsf{EvidenceOK}(s)$ for acceptance of policy-compliant evidence matching $b_s$, $\mathsf{TLSOK}(s)$ for successful ordinary TLS validation, and $\mathsf{Release}(s,p)$ for the release of protected application data carrying $p$.

The argument uses four explicit assumptions.
\begin{description}
\item[A1: Evidence and binding.] Accepted evidence authenticates the approved program and its report data. Forging evidence or finding distinct accepted inputs with the same binding hash has negligible probability.
\item[A2: Protected provenance and state.] The approved observer endorses only a share generated and used by the protected TLS endpoint for the corresponding local flow. That endpoint generates fresh private shares, validates peer shares, and confines the private share, derived receiving keys, and equivalent decryption capability to the protected boundary. The adversary has neither copied receiving state nor a policy-accepted endpoint that discloses request plaintext to it. This is the implementation and deployment premise supplied by the trusted TLS stack, kernel, and observer.
\item[A3: Session-local gating.] The client compares evidence with the share in its own negotiated handshake. It releases protected application data only after both evidence verification and ordinary TLS validation succeed, and fails closed on mismatch or failure:
\[
\begin{aligned}
\mathsf{Release}(s,p)\Rightarrow{}&\mathsf{EvidenceOK}(s)\\
 &{}\land\mathsf{TLSOK}(s).
\end{aligned}
\]
ClientHello may proceed before this event. The TLS stack remains responsible for certificate, transcript, and Finished checks.
\item[A4: TLS key and record security.] For fresh, validated ECDHE shares whose private and derived secrets remain confined, the TLS~1.3 key schedule and record protection preserve request confidentiality, including when the reusable certificate signing credential is disclosed. This assumption concerns secrecy of the resulting channel with the bound protected share; certificate-based peer authentication alone is insufficient once that credential is compromised.
\end{description}

\subsubsection{Binding and Disclosure Properties}
\begin{lemma}[Accepted-share provenance]
\label{lem:provenance}
Under A1--A3, any released request is encrypted in a TLS session whose server share is the share endorsed by a policy-accepted protected endpoint, except upon evidence forgery or binding-hash failure.
\end{lemma}
\begin{proof}
Release implies acceptance of evidence matching $b_s$ by A3. Excluding evidence forgery, that report was produced by the approved program. By A2, its binding input contains a locally generated and used protected share $Y^*$. Excluding a hash collision, injective encoding and equality of the verified bindings imply equality of role, group, and share; hence $Y_s=Y^*$. The comparison is with the client's own handshake, so evidence for a different share cannot authorize this release.
\end{proof}

\begin{proposition}[Credential-only ESR resistance]
\label{prop:credential}
Under A1--A4, possession of the reusable service signing credential alone does not enable the receiving-state copying ESR. In particular, for a uniformly sampled $\ell$-bit secret request independent of public information and authorized output,
\[
\begin{aligned}
&\Pr[\mathsf{Accept}_C(s)\land\mathsf{Release}(s,p)
      \land\mathsf{Recover}_A(p)]\\
&\qquad\leq 2^{-\ell}+\operatorname{negl}(\lambda),
\end{aligned}
\]
where $\lambda$ is the security parameter. For general requests, confidentiality is understood relative to the information already available from public metadata and authorized output.
\end{proposition}
\begin{proof}
Condition on the absence of the negligible evidence and binding failures in Lemma~\ref{lem:provenance}. If the adversary supplies a server share for which it independently knows the receiving secret, that share cannot satisfy the protected-origin and confinement premise A2; thus it cannot pass the evidence check for the released session. If it instead forwards the endorsed protected share, its copied signing credential supplies neither the ECDHE private value nor the resulting receiving keys. A4 then protects the released request. Relaying the genuine handshake and ciphertext preserves delivery to the protected receiver without providing plaintext to the relay. Recovery therefore requires a primitive failure or guessing the independent secret. Adding the excluded negligible failure probabilities yields the bound.
\end{proof}

The proposition composes a trusted provenance invariant with TLS confidentiality; it is not a proof that packet observation alone establishes key confinement. Finished supplies ordinary TLS key confirmation, rather than a standalone proof of knowledge of $y_s$. A copied ephemeral secret or an adversary-controlled accepted environment changes the resource premises and returns to the corresponding ESR construction in Section~\ref{sec:esr}.

\subsubsection{Mutual Attestation and Batch}
For mutual attestation, the client share uses the distinct binding
\[H(\operatorname{Encode}(\mathsf{client},\mathsf{group}_s,X_s)).\]
The same provenance lemma applies with the endpoint roles exchanged. For data in either direction, the disclosure proposition requires verification and gating before release in that direction, together with confinement of that receiver's secrets. Role separation prevents a server-share endorsement from substituting for a client-share endorsement.

Batch establishes a different property. If a verifier accepts the challenge-bound digest of a canonical manifest, A1 and correct manifest processing imply that the approved observer matched the manifest's share observations. This is delayed evidence of those observations: business requests may already have completed. Its proof obligation is manifest authenticity and coverage, rather than the pre-release implication in A3. Passive share observations alone do not establish Finished validation or complete local TLS termination.

\subsubsection{Connection Lifetime and Scope}
The proposition covers full ECDHE handshakes with resumption and 0-RTT disabled. Requests and TLS KeyUpdate operations on an accepted connection inherit that connection's decision; they do not refresh attestation. Extending the result to resumption requires binding the resumed receiving capability to an accepted protected endpoint and defining its lifetime. Early data additionally requires a release rule and replay treatment before any fresh server share is available.

\endgroup

\subsection{Mutual Session Attestation}
\label{sec:mutual}

Confidential cloud services may span multiple protected service nodes, so an inter-service TLS connection can require both endpoints to verify that the peer runs in an accepted protected environment. TLSLatch extends session attestation to both sides of such a connection. The client-side observer registers the locally generated client share $X=g^x$ and produces evidence with
client-role domain separation. While the TLS handshake proceeds, the
client sends this evidence to the server. After verifying this binding, the server produces its own evidence over \(Y\), incorporating the verified client binding. The client
verifies this combined binding before opening its gate.

The two attestation operations overlap the original handshake and require no additional client-challenge round trip.

%The two attestation operations overlap the original handshake and require no additional client-challenge round trip. As in the server-only design, each endpoint must verify its peer before releasing sensitive data; mutual evidence alone does not enforce this ordering. An endpoint that may transmit protected data first therefore requires its own outbound gate.
%The two attestation operations can overlap the original handshake and
%do not require an additional client-challenge round trip. As in the
%server-only case, replaying a public share and its evidence does not
%supply the corresponding ephemeral receiving state. TLS Finished carries
%the authenticated key exchange into the accepted channel in each
%direction. Mutual evidence does not by itself determine when sensitive
%data may be released: an endpoint that sends protected data before
%verifying its peer requires the corresponding outbound gate.

\subsection{Delayed Batch Confirmation}
\label{sec:batch}

%TLSLatch-Batch is a distinct operating mode that trades the
%pre-delivery guarantee of the main protocol for amortized attestation.

Some users of established confidential-cloud providers may accept delayed verification in exchange for lower online overhead. TLSLatch-Batch targets this tradeoff by amortizing attestation across multiple connections.

Instead of gating each connection, both sides passively record
ServerHello observations. At a configured interval, the client submits
a batch identifier, a fresh challenge, and a canonical deduplicated list
of record digests. The server matches these entries against its own
observations, prevents an observation from being committed to multiple
batches, and attests a domain-separated digest of the resulting batch.
The client recomputes the digest and verifies the evidence.

Because Batch removes the send gate, requests may be released before confirmation completes. It therefore provides delayed confirmation that an attested program observed and checked the recorded connections, rather than the pre-delivery guarantee of TLSLatch. The configured interval controls attestation frequency and confirmation delay, while the underlying TLS sessions retain their normal key lifetimes.

%Because Batch removes the send gate, requests may already have been
%released when confirmation completes. TLSLatch-Batch provides delayed confirmation that an attested program observed and checked the recorded connections. The configured interval controls attestation frequency and confirmation delay. Batch therefore provides delayed confirmation that an attested
%program observed and checked the submitted connections, rather than the
%pre-delivery receiver guarantee established above. Its configurable
%interval controls attestation frequency and confirmation delay, not the
%lifetime of the underlying TLS keys.
% !TeX root = ../main.tex
\section{Implementation}
\label{sec:implementation}

We implement TLSLatch as user-space agents, with server observation and report generation in a Linux Hygon CSV CVM and clients on Linux, Windows, and macOS. Applications retain their existing TLS
libraries, certificates, and protocols. The implementation realizes the
two mechanisms of Section~\ref{sec:design}: associating an observed TLS
share with an active local endpoint, and temporarily withholding client
traffic until verification completes.

To improve deployability across platforms and system configurations, we realize the client-side gate using existing OS traffic-control interfaces: nftables/NFQUEUE on Linux, WinDivert on Windows, and Network Extension on macOS. The required primitives are available on Linux 3.14+, Windows 10/11 and Windows Server 2016+, and macOS 10.15+, respectively.
%The agents build on mature OS traffic-control interfaces: nftables/NFQUEUE on Linux, WinDivert on Windows, and Network Extension on macOS. The required primitives are available on Linux 3.14+, Windows 10/11 and Windows Server2016+, and macOS 10.15+, respectively.

%We deliberately build on existing packet- and flow-control facilities
%rather than requiring new kernel code. In particular, the Linux path
%does not depend on eBPF-specific hooks, and Windows uses WinDivert's
%existing driver instead of a custom WFP callout. These choices avoid
%additional kernel-version and driver-signing constraints. Because
%TLSLatch leaves the interception path after authentication, lower-level
%alternatives affect only connection setup; in our prototypes they did
%not materially improve end-to-end performance.

\paragraph{Server observation on Linux.}
The server passively observes ClientHello and locally emitted ServerHello messages through direction-filtered libpcap capture. Each server share \(Y\) is associated with the active local socket that emitted it and stored as a connection-scoped observation. The control interface can only reference an existing observation; it never accepts a caller-supplied \(Y\) for attestation. Before generating evidence, the agent revalidates the associated connection, ensuring that the reported share still belongs to the active local TLS endpoint. This realizes the provenance requirement of Section 4.1 without instrumenting the TLS library.

%The control request identifies the corresponding observation using a
%domain-separated hash of the client key-share encoding. The agent
%resolves this handle only against an existing local record and never
%accepts a caller-supplied $Y$ for attestation. Because the control
%request and ServerHello may arrive in either order, report generation
%waits until both are present and revalidates the associated connection
%before using the record. These checks realize the provenance contract
%of Section~\ref{sec:binding} without instrumenting the TLS library.

\paragraph{One-time client gate.}
The Linux client uses nftables and NFQUEUE for temporary
release control. ClientHello is allowed to proceed, while subsequent
outbound packets are retained pending the attestation decision. After
successful verification, queued packets are released and authorization
is committed to the kernel connection-tracking state. Packets belonging
to that connection then bypass NFQUEUE entirely, and the authorization
expires with the connection. Failed verification drops the pending
traffic and leaves the destination blocked. Thus user-space packet
processing is confined to connection establishment rather than the
application data path.

\paragraph{Windows and macOS}
%The other clients implement the same gate semantics with platform-native
%interfaces rather than reproducing the Linux packet path.
%Windows uses the WinDivert user-space API and its existing driver
%~\cite{windivertdocs}: packets are diverted and buffered while
%verification is pending, then reinjected before interception is removed
%on success. macOS uses a Network Extension content filter
%~\cite{applefilter}. The operating system provides a flow identifier and
%can suspend outbound flow data until a decision is available, so the
%implementation resumes or rejects the flow rather than reinjecting
%individual packets. Table~\ref{tab:platforms} summarizes these
%platform-specific substitutions. The listed TLS stacks are those used
%in our evaluation; TLSLatch does not call into their internal TLS APIs.
%Detailed parsing, race handling, and platform lifecycle code are
%included in the released implementation.
Windows and macOS realize the same one-time gate using the platform interfaces introduced above. Both operate outside the TLS library and remove user-space interception after authorization. Table~\ref{tab:platforms} summarizes the platform-specific mechanisms.

% !TeX root = ../main.tex
\begin{table*}[t]
	\caption{Platform-specific realizations of the same client-side gate.
		Server observation and report generation run in the Linux CVM.}
	\label{tab:platforms}
	\centering
	\setlength{\tabcolsep}{4pt}
	\renewcommand{\arraystretch}{1.08}
	
	\begin{tabularx}{\textwidth}{
			@{}
			>{\raggedright\arraybackslash}p{3.0cm}
			*{3}{>{\raggedright\arraybackslash}X}
			@{}}
		\toprule
		Mechanism & Linux & Windows & macOS \\
		\midrule
		
		Evaluated TLS stack
		& OpenSSL
		& Schannel
		& Network.framework \\
		
		Gate interface
		& nftables + NFQUEUE
		& WinDivert user-space API
		& Network Extension filter \\
		
		During verification
		& Queue outbound packets
		& Buffer diverted packets
		& Pause outbound flow data \\
		
		After success
		& Release; bypass via conntrack mark
		& Reinject; stop diversion
		& Resume flow \\
		
		Flow association
		& Local socket + connection tracking
		& TCP flow + initial sequence
		& OS flow identifier \\
		
		Failure
		& Drop pending traffic; block target
		& Keep diversion until isolation
		& Deny flow \\
		
		Batch capture
		& libpcap
		& Npcap
		& libpcap \\
		
		\bottomrule
	\end{tabularx}
	
	\par\vspace{2pt}
	\raggedright
	The TLS-stack row reports the applications used in evaluation rather
	than a TLSLatch dependency.
\end{table*}

\paragraph{Mutual attestation and Batch.}
Mutual attestation reuses the same observation and report-generation path at both Linux endpoints, with evidence generation triggered by local handshake observations so that it overlaps the TLS exchange. TLSLatch-Batch instead records handshakes using libpcap on Linux and macOS and Npcap on Windows. A periodic timer batches accumulated observations into a single report request, moving report generation and verification off the per-connection critical path.
% !TeX root = ../main.tex
\section{Performance Evaluation}
\label{sec:evaluation}
We evaluate TLSLatch across Linux, Windows, and macOS, focusing on
end-to-end latency, payload-path overhead, protection-component
resource cost and scalability, and the overhead of mutual attestation
and delayed Batch confirmation.

\subsection{Experimental Setup}

The server runs openEuler in a Hygon CSV CVM with 16 vCPUs and 32 GB of memory, with attestation reports generated by physical Hygon hardware. The Linux client runs in a second CSV CVM with 4 vCPUs and 4 GB of memory. The Windows client is a physical Windows 11 machine with 8 CPUs and 16 GB of memory, while the macOS client is a physical Apple M4 system with 16 GB of memory running macOS 26. Linux, Windows, and macOS use OpenSSL, Schannel, and Network.framework, respectively.

The Linux, Windows, and macOS client–server paths have mean RTTs of 0.257, 1.214, and 0.577 ms, respectively, with single-stream TCP goodput of 12.91 Gbps, 672.1 Mbps, and 925.9 Mbps. We therefore compare schemes within each client platform.

We compare Native TLS, TLSLatch, TLSLatch-Batch, TLS+RA, TNG in single-tunnel and nested-TLS modes, and CMaaS, following their public implementations while adapting attestation-dependent components to Hygon CSV. The evaluated TNG and CMaaS configurations run as separate user-space components outside the application TLS stack and are evaluated across all three client platforms; TLS+RA modifies the TLS stack and is therefore evaluated only on Linux, where we use OpenSSL. CMaaS is evaluated in its steady-state mode with an established application key, including its AES-GCM encryption and Base64 encoding.

Latency experiments establish fresh full TLS~1.3 connections with
resumption disabled and transfer 1~KB, 1~MB, 16~MB, or 64~MB.
Completion time spans connection initiation through receipt of the
server application-level acknowledgment. We report means after removing the five shortest and five longest
measurements from each block of 50 requests. For resource measurements, we use \emph{guard} to denote the scheme-specific protection components and required helpers measured separately from the business application.
%For payload and memory sizes, KB and MB denote $2^{10}$ and
%$2^{20}$ bytes, respectively.

\subsection{End-to-End Latency and Transfer Overhead}
\begin{figure*}[t]
\centering\includegraphics[width=\textwidth]{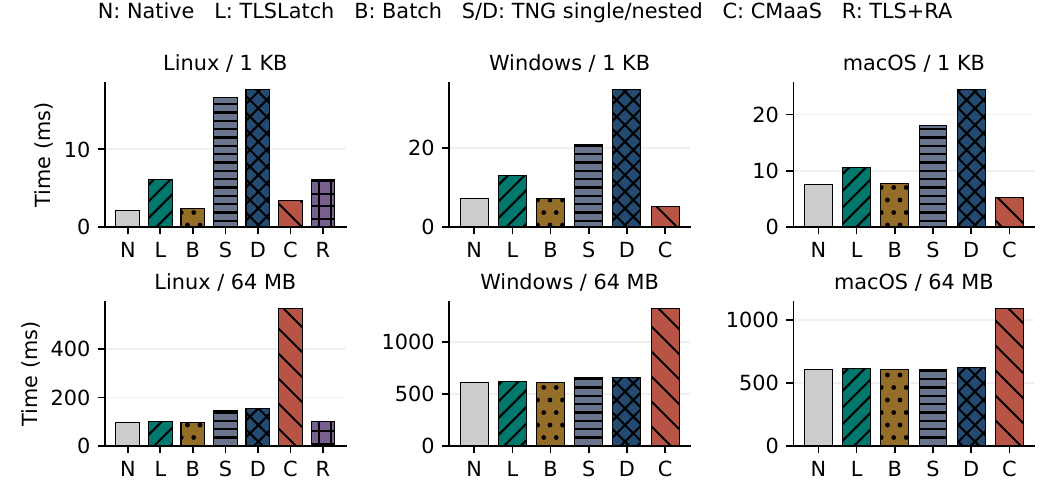}
\caption{Trimmed mean end-to-end completion time for
	setup-dominated 1~KB requests and 64~MB transfers.
	TLS+RA is evaluated on Linux. CMaaS uses an established
	application key; TLSLatch-Batch reports request completion
	before delayed confirmation.}
\Description{Six panels compare one KB and 64 MB request completion times across Linux, Windows and macOS.}
\label{fig:latency-linear}
\end{figure*}

%Figure~\ref{fig:latency-linear} shows that TLSLatch primarily
%adds a fixed authentication cost during connection setup. For
%1~KB requests, TLSLatch completes in 6.09, 13.04, and 10.53~ms
%on Linux, Windows, and macOS, respectively, adding 2.89--5.84~ms
%over native TLS. On Linux, TLSLatch and TLS+RA take 6.09 and
%6.10~ms, respectively, showing that transparent integration adds
%little latency beyond direct TLS integration. Compared with nested
%TNG, TLSLatch reduces completion time by 56.9--65.5\% across the
%three platforms. CMaaS is faster for these small requests because
%its steady-state path reuses an established application key; initial
%key establishment costs 5.78, 10.86, and 8.24~ms on Linux,
%Windows, and macOS. TLSLatch-Batch, which moves confirmation off
%the request critical path, remains within 0.24~ms of native TLS at
%1~KB on all three platforms.

Figure~\ref{fig:latency-linear} shows that TLSLatch primarily adds a fixed authentication cost during connection setup. For 1 KB requests, TLSLatch completes in 6.09, 13.04, and 10.53 ms on Linux, Windows, and macOS, respectively, adding 2.89–5.84 ms over Native TLS. On Linux, TLSLatch closely matches TLS+RA (6.09 vs. 6.10 ms), indicating that transparent integration adds little latency beyond direct TLS integration. Compared with nested TNG, TLSLatch reduces completion time by 56.9–65.5\% across the three platforms. CMaaS is faster for these small requests because its steady-state path reuses an established application key; initial key establishment takes 5.78–10.86 ms across platforms. TLSLatch-Batch remains within 0.24 ms of Native TLS at 1 KB on all three platforms.

As payload size grows, TLSLatch approaches Native TLS because the gate leaves the application data path after authentication. At 64 MB, completion time is only 4.6\%, 1.7\%, and 0.4\% above Native TLS on Linux, Windows, and macOS, respectively. CMaaS continues to encrypt and encode application payloads, and TLSLatch completes the same transfers 82.4\%, 53.0\%, and 43.8\% faster. Together with the 1 KB results, this scaling shows that TLSLatch introduces predominantly per-session authentication cost rather than payload-proportional processing.

%As payload size grows, TLSLatch approaches native TLS because the
%agent leaves the application data path after authentication. At
%64~MB, completion time is 99.49, 620.63, and 610.47~ms on Linux,
%Windows, and macOS, only 4.6\%, 1.7\%, and 0.4\% above native TLS.
%CMaaS continues to encrypt and encode application payloads, and
%TLSLatch completes the same transfers 82.4\%, 53.0\%, and 43.8\%
%faster, respectively. The short- and large-transfer results therefore
%show a predominantly per-session authentication cost rather than
%payload-proportional processing.

A Linux microbenchmark further localizes this fixed cost. CSV report
generation takes 3.369~ms on average, while evidence verification,
ClientHello registration, and gate authorization and release take
0.243, 0.195, and 0.114~ms, respectively. These operations may
overlap and therefore do not form an additive decomposition of
end-to-end latency.

\begin{figure}[t]
\centering\includegraphics[width=\columnwidth]{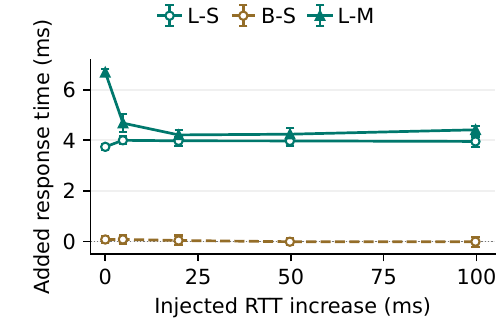}
\caption{Response-time increments over Native TLS.
L-S/L-M: server-only/mutual TLSLatch; B-S: server-only Batch.
Paired trimmed means with 95\% block-bootstrap intervals.
Injected RTT sums both one-way delays; Batch excludes later confirmation.}
\Description{Three curves compare server-only TLSLatch, server-only Batch, and mutual TLSLatch response-time increments as injected RTT increases.}
\label{fig:rtt-sensitivity}
\end{figure}

\paragraph{RTT sensitivity.}
Figure~\ref{fig:rtt-sensitivity} varies the network RTT for sequential 1 KB requests, applying the same delay to business and attestation traffic. With server-only attestation, TLSLatch adds 3.74–4.00 ms over Native TLS across measured RTTs of 0.312–100.105 ms. With mutual attestation, the increment is 6.69 ms at baseline but falls to 4.21–4.68 ms under nonzero injected delay and remains nearly constant thereafter. This behavior is consistent with client evidence generation overlapping handshake propagation rather than introducing an additional network round trip.

\subsection{Resource Cost and Scalability}
\begin{figure*}[t]
\centering\includegraphics[width=\textwidth]{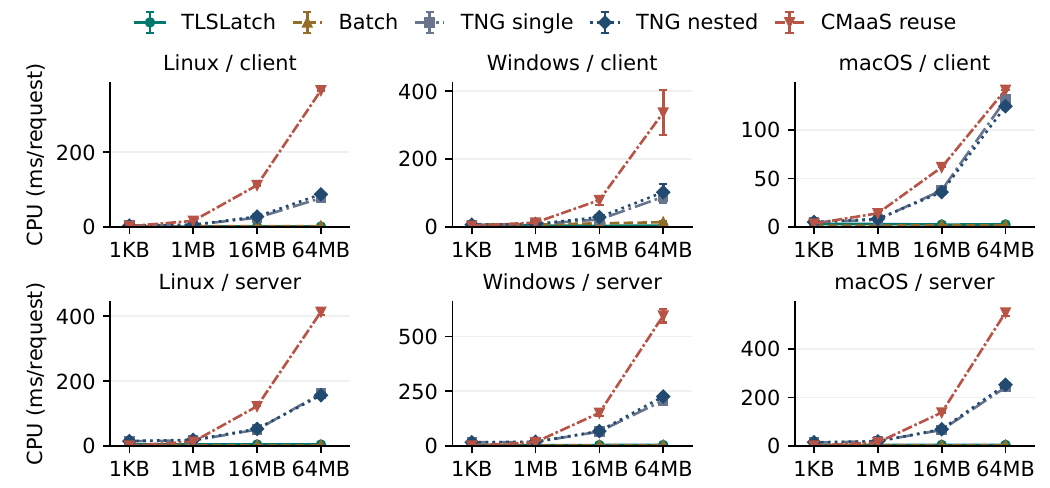}
\caption{Independent guard CPU per request, excluding the business application. Points average two resource batches; bars show their minimum and maximum. TLSLatch CPU remains largely independent of payload size. TLS+RA is omitted because its attestation code is embedded in the application.}
\Description{Six panels show client and server guard CPU across four upload sizes on each client platform.}
\label{fig:cpu}
\end{figure*}

Figure~\ref{fig:cpu} shows that TLSLatch’s guard CPU cost remains largely independent of payload size. Across the four payload sizes from 1 KB to 64 MB, client CPU ranges from 1.09–1.24 ms on Linux, 3.39–4.43 ms on Windows, and 2.48–3.09 ms on macOS, while the server guard remains at approximately 4.6–5.2 ms. In contrast,
mechanisms that remain on the payload path accumulate work as
transfers grow: at 64~MB, TNG-single consumes 77.27, 88.67, and
131.31~ms of client CPU across the three platforms, while CMaaS
consumes 365.06, 337.11, and 141.42~ms. This trend is consistent
with the end-to-end results in Figure~\ref{fig:latency-linear}.

\begin{table}[t]
\caption{Guard memory for 64 MB uploads (MB). C/S denotes client/server. Values are two-batch means of sampled component-set peaks, excluding business processes. Windows uses Working Set; Linux and macOS use RSS.}
\label{tab:resources}
\centering\setlength{\tabcolsep}{4pt}
\begin{tabular}{@{}llrrr@{}}\toprule
Scheme & & Linux & Windows & macOS \\ \midrule
\multirow{2}{*}{\sys{}} & C & 5.69 & 22.07 & 11.15 \\
 & S & 17.92 & 13.89 & 12.02 \\
\addlinespace[2pt]
\multirow{2}{*}{TLSLatch-Batch} & C & 18.77 & 29.05 & 14.78 \\
 & S & 19.78 & 19.78 & 21.24 \\
\addlinespace[2pt]
\multirow{2}{*}{TNG single} & C & 31.56 & 24.81 & 21.81 \\
 & S & 35.98 & 39.60 & 37.71 \\
\addlinespace[2pt]
\multirow{2}{*}{TNG nested} & C & 32.04 & 24.93 & 21.86 \\
 & S & 38.05 & 39.69 & 39.41 \\
\addlinespace[2pt]
\multirow{2}{*}{CMaaS reuse} & C & 700.67 & 566.27 & 440.26 \\
 & S & 398.38 & 398.38 & 398.72 \\
\bottomrule\end{tabular}
\end{table}

Table~\ref{tab:resources} reports the guard-memory footprint for 64 MB transfers. TLSLatch remains within 5.69–22.07 MB on the client side and 12.02–17.92 MB on the server, whereas the evaluated CMaaS implementation reaches 440.26–700.67 MB and about 398 MB, respectively, under its buffering and payload-transformation path.

% Replace the input path below with the current filename defining
% Figure~\ref{fig:scalability}, if different.
\begin{figure*}[t]
\centering\includegraphics[width=\textwidth]{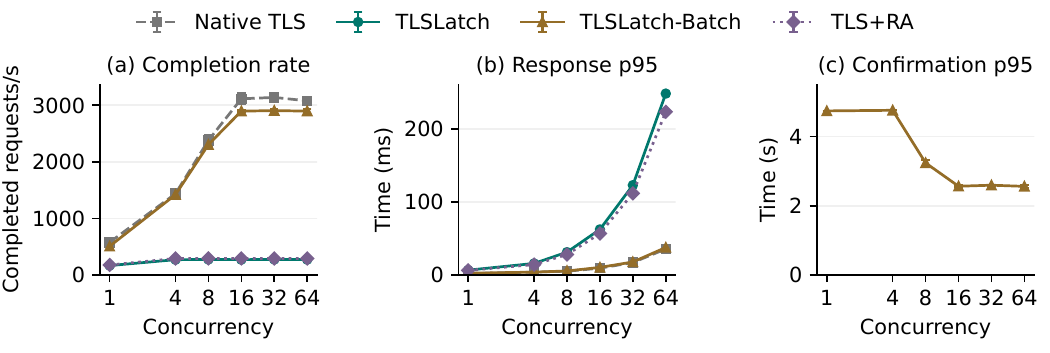}
\caption{Batch approaches native business throughput, while both
per-session attestation schemes saturate on their single-report paths.
(a) Mean business-completion throughput with two-run ranges.
(b) Response p95 pooled over 16,384 connections per configuration.
(c) Batch confirmation: mean and range of the two run-level p95 values,
including final flush. Concurrency results retain all samples;
Batch releases traffic before its later confirmation.}
\Description{Four schemes are compared on business throughput and
response p95; a third panel reports delayed Batch confirmation.}
\label{fig:auth-scalability}
\end{figure*}

%Figure~\ref{fig:auth-scalability} compares Native TLS, TLSLatch,
%TLSLatch-Batch, and TLS+RA under concurrent 1~KB requests on Linux.
%The four schemes use the same OpenSSL and application workload,
%with two runs of 8,192 measured connections per configuration.
%TLSLatch and TLS+RA each serialize report generation and
%plateau near 270 and 291 completed requests/s, respectively.
%At concurrency 64, TLSLatch reaches 271.34 requests/s,
%6.6\% below TLS+RA's 290.65 requests/s; their response p95 values
%are 248.17 and 223.19~ms. The shared plateau, together with the
%3.369~ms report-API measurement, identifies synchronous per-session
%report generation as the principal bottleneck in these implementations.

Figure~\ref{fig:auth-scalability} compares Native TLS, TLSLatch, TLSLatch-Batch, and TLS+RA under concurrent 1 KB requests on Linux. TLSLatch and TLS+RA both saturate on synchronous per-session report generation, plateauing near 270–290 requests/s. At concurrency 64, TLSLatch reaches 271.34 requests/s, 6.6\% below TLS+RA’s 290.65 requests/s, with response p95 of 248.17 and 223.19 ms, respectively. Together with the 3.369 ms report-generation microbenchmark, these results identify synchronous report generation as the dominant scalability bottleneck.

TLSLatch-Batch amortizes this cost across connections. At concurrency 64, it reaches 2,894.11 requests/s, or 94.1\% of Native TLS, with response p95 of 37.58 ms versus 35.55 ms. Figure~\ref{fig:auth-scalability}(c) separately reports delayed confirmation latency. Batch therefore approaches native business throughput for workloads that accept delayed verification, while TLSLatch retains pre-delivery verification.

\begin{table}[t]
\caption{Batch amortization and confirmation. Connections/report includes
warm-up and final batches. Delay is the mean time from the oldest
observation in a periodic batch to successful verification.}
\label{tab:batch-confirmation}
\centering\setlength{\tabcolsep}{4pt}
\begin{tabular}{@{}lrr@{}}\toprule
Client & Connections/report & Delay (s) \\ \midrule
Linux & 18.46 & 4.49 \\
Windows & 10.74 & 4.37 \\
macOS & 11.43 & 4.74 \\
\bottomrule\end{tabular}
\end{table}

Table~\ref{tab:batch-confirmation} quantifies the corresponding amortization–delay tradeoff. With a 5-s reporting interval, each report covers 10.7–18.5 connections on average across the three platforms, while the oldest observation in a batch is confirmed after 4.37–4.74 s. Together with the near-native request latency and throughput above, these results show that Batch shifts attestation cost and confirmation off the per-connection critical path at the expense of delayed verification.

\subsection{Mutual Attestation}

\begin{table}[t]
\caption{Mutual-attestation trimmed mean completion time (ms). Requests and responses each have the indicated size.}
\label{tab:mutual}
\centering\setlength{\tabcolsep}{4pt}
\begin{tabular}{@{}lrrrr@{}}\toprule
Scheme & 1 KB & 1 MB & 16 MB & 64 MB \\ \midrule
Native TLS & 2.19 & 6.44 & 54.52 & 193.56 \\
TLSLatch server-only & 6.10 & 11.13 & 62.83 & 200.42 \\
TLSLatch mutual & 8.88 & 13.94 & 65.74 & 207.93 \\
TLS+RA mutual & 10.33 & 15.46 & 62.09 & 206.05 \\
TNG mutual nested & 30.80 & 37.36 & 107.70 & 322.28 \\
\bottomrule\end{tabular}
\end{table}

Table~\ref{tab:mutual} evaluates mutual session attestation between two CSV CVMs. For 1 KB requests and responses, mutual TLSLatch completes in 8.88 ms, adding 2.78 ms over server-only TLSLatch. Its completion time is 14.0\% lower than mutual TLS+RA and 71.2\% lower than nested TNG. The added cost over server-only TLSLatch remains 2.78–2.91 ms through 16 MB; at 64 MB in each direction, mutual TLSLatch completes in 207.93 ms, only 3.7\% above server-only TLSLatch. This scaling indicates that mutual-attestation overhead is primarily per-connection rather than payload-proportional.

\FloatBarrier
% !TeX root = ../main.tex
\section{Conclusion}
Confidential cloud services must ensure that the endpoint receiving plaintext is the protected environment they intend to trust. We introduced ESR to characterize this problem and presented TLSLatch, a transparent session-attestation mechanism that binds the current TLS session to an attested receiving endpoint. Our cross-platform evaluation shows that this binding can be added transparently to existing TLS deployments with low overhead. More broadly, TLSLatch provides a general way to strengthen endpoint binding for TLS-based trusted and confidential computing services.

\FloatBarrier
\bibliographystyle{ACM-Reference-Format}
\bibliography{refs}
\end{document}